\documentclass[conference]{IEEEtran}
\usepackage{graphics,color}
\usepackage{amsmath,amssymb,epsfig,dsfont}
\usepackage{verbatim,subfigure}
\IEEEoverridecommandlockouts

\begin{document}
\title{Collecting Coded Coupons over Overlapping Generations\\[-3mm]}
\author{
\IEEEauthorblockN{Yao Li}
\IEEEauthorblockA{WINLAB, ECE Department\\ Rutgers University\\
 yaoli@winlab.rutgers.edu} \and
\IEEEauthorblockN{Emina Soljanin}
\IEEEauthorblockA{Bell Labs\\ Alcatel-Lucent\\
emina@alcatel-lucent.com} \and \IEEEauthorblockN{Predrag
Spasojevi\'{c}}
\IEEEauthorblockA{WINLAB, ECE Department\\ Rutgers University\\
spasojev@winlab.rutgers.edu} 

 \thanks{
This work is supported by NSF grant No. CNS 0721888.} }
\maketitle

\begin{abstract}
Coding over subsets (known as generations) rather than over all
content blocks in P2P distribution networks and other applications
is necessary for a number of practical reasons such as computational
complexity. A penalty for coding only within generations is an
overall throughput reduction. It has been previously shown that
allowing contiguous generations to overlap in a head-to-toe manner
improves the throughput. We here propose and study a scheme,
referred to as the {\it random annex code}, that creates shared
packets between any two generations at random rather than only the
neighboring ones. By optimizing very few design parameters, we
obtain a simple scheme that outperforms both the non-overlapping and
the head-to-toe overlapping schemes of comparable computational
complexity, both in the expected throughput and in the rate of
convergence of the probability of decoding failure to zero. We
provide a practical algorithm for accurate analysis of the expected
throughput of the random annex code for finite-length information.
This algorithm enables us to quantify the throughput vs.\
computational complexity tradeoff, which is necessary for optimal
selection of the scheme parameters.
\end{abstract}

\newtheorem{theorem}{Theorem}
\newtheorem{lemma}[theorem]{Lemma}
\newtheorem{claim}[theorem]{Claim}
\newtheorem{corollary}[theorem]{Corollary}
\newtheorem{remark}{Remark}

\section{Introduction}

In randomized network coding applications, such as P2P content
distribution \cite{avalanche} or streaming, the source splits its
content into blocks. For a number of practical reasons (e.g,
computational complexity and delay reduction, easier
synchronization, simpler content tracking), these blocks are further
grouped into subsets referred to as generations, and, in the coding
process at the source and within the network, only packets in the
same generation are allowed to be linearly combined. A penalty for
coding only within generations is an overall throughput reduction.
The goal of this work is to develop a strategy which allows
generations to overlap, and hence, to improve the throughput while
maintaining the benefits brought up by introducing generations.

Since the idea of coding over generations has been introduced by
Chou et al.\ in \cite{choupractical}, a number of issues concerning
such coding schemes have been addressed. Maymounkov et al.\ studied
coding over generations with {\em random scheduling} of the
generation transmission in \cite{petarchunked}. They referred to
such codes as \emph{chunked codes}. This random scheduling coding
scheme has also been studied in our recent work
\cite{nonOverlapping}. The tradeoff between the reduction in
computational complexity and the throughput benefits was addressed
in \cite{petarchunked,nonOverlapping}, and the tradeoff between the
reduction in computational complexity and the resilience to peer
dynamics in P2P content distribution networks was the topic of
\cite{niudiresilience}.

Here, we are particularly interested in recovering some of the
throughput that is lost as a consequence of coding over generations
by allowing the generations to overlap. It has been observed that
with random scheduling, some generations accumulate their packets
faster, and can be decoded earlier than others. If generations are
allowed to overlap, the decoding of some generations will reduce the
number of unknowns in those lagging behind but sharing packets with
those already decoded, and thus help them become decodable with
fewer received coded packets, consequently improving the throughput.
Coding with \emph{overlapping} generations was first studied in
\cite{queensoverlap} and \cite{carletonoverlap}.

In this work, we propose a coding scheme over generations that share
randomly chosen packets. We refer to this code as the \emph{random
annex code}. 
We demonstrate with both a heuristic analysis and simulations that,
assuming comparable computational complexity, with its small number
of design parameters optimized, the simple random annex coding
scheme outperforms the non-overlapping scheme, as well as a
``head-to-toe'' overlapping scheme. The head-to-toe scheme was the
topic of \cite{carletonoverlap}, in which only contiguous
generations overlap in a fixed number of information packets in an
end-around mode, and hence only contiguous generations can benefit
from each other in decoding.

The main contribution of this paper is an accurate and practical
evaluation of the expected throughput performance of the random
annex code for finite information lengths. We first anatomize the
overlapping structure of the generations to quantify the benefit of
previously decoded generations on those not yet decoded, namely, how
many fewer coded packets are needed for each generation to be
decoded. With this accomplished, we are then able to place the
overlaps in oblivion and analyze the throughput performance of
random annex codes under the coupon collection setting. We succeed
to optimize the scheme design parameters based on the evaluation of
the overhead  necessary to decode all the packets.


Our paper is organized as follows: In Section \ref{sec:model}, we
describe the coding scheme for the random annex code, and introduce
the coupon collector's model under which our analysis for code
throughput is studied. In Section \ref{sec:analysis}, we analyze the
overlapping structure of the random annex code, and give a practical
algorithm to evaluate the expected code performance for finite
information length. In Section \ref{sec:results}, we present our
numerical evaluation and simulation results which illustrate and
quantify the throughput vs.\ computational complexity tradeoff
brought up by our scheme. Section \ref{sec:conclusion} concludes.

\section{Coding over Generations} \label{sec:model}

\subsection{Coding over Overlapping Generations in Unicast}

We refer to the code we study as the \emph{random annex code}. This
section describes (a) the way the generations are formed,
(b) the encoding process, (c) the decoding algorithm, and (d)
how the computational complexity for the random annex code is
measured.

\setcounter{paragraph}{0}
\paragraph{Forming Overlapping Generations}\label{par:formgen}
We first divide file $\mathcal{F}$ into $N$ information packets $p_1, p_2, \dots,p_N$.
Each packet $p_i$ is represented as a column vector of $d$ information symbols in Galois
field $GF(q)$. We then form $n$ overlapping generations in two steps as follows:
\begin{enumerate}
  \item We partition the $N$ packets into sets $B_1,B_2, \dots, B_n$ of
  each containing $h$ consecutive packets. We refer to these $n=N/h$
  sets as \textit{base generations}.
Thus, $B_i=\{p_{(i-1)h+1},p_{(i-1)h+2},\dots,p_{ih}\}$ for
$i=1,2,\dots,n$. We assume that $N$ is a
multiple of $h$ for convenience. In practice, if $N$ is not a multiple of $h$,  we
set $n=\lceil N/h \rceil$ and assign the last $[N-(n-1)h]$ packets to
the last (smaller) base generation.
  \item To each base generation $B_i$, we add a random \emph{annex} $R_i$, consisting of
$l$ packets chosen uniformly at random (without replacement) from
the $N-h=(n-1)h$ packets in $\mathcal{F}\backslash B_i$. The base generation together
with its annex constitutes the
\emph{extended generation} $G_i=B_i\cup R_i$. The size of each $G_i$ is $g=h+l$.
Throughout this paper, unless otherwise stated, the term ``generation'' will refer to ``extended generation''
whenever used alone.
\end{enumerate}

The members of $G_i$ are enumerated as
$p^{(i)}_1,p^{(i)}_2,\dots,p^{(i)}_g.$

\paragraph{Encoding} The encoding process is oblivious to overlaps between
generations. In each transmission, the source first selects one of
the $n$ generations with equal probability. Assume $G_j$ is chosen.
Then the source chooses a coding vector $e=[e_1,e_2,\dots,e_{g}]^T$
with each entry chosen independently and equally probably from
$GF(q)$. A new packet $\bar{p}$ is then formed by linearly combining
packets from $G_j$ by $e$: $\bar{p}=\sum_{i=1}^{g}e_ip^{(j)}_{i}$.
The coded packet $\bar{p}$ is then sent over the communication link
to the receiver along with the coding vector $e$ and the generation
index $j$.

\paragraph{Decoding}
Decoding starts with a generation for which the receiver has
collected $(h+l)$ coded packets with linearly independent coding
vectors. The information packets making up this generation are
decoded by solving a system of $(h+l)$ linear equations in $GF(q)$
formed by the coded packets and the linear combinations of the
information packets by the coding vectors. Each decoded information
packet is then removed as an unknown from other generations' coded
packets in which it participates. Consequently, the number of
unknowns in all generations overlapping with those that are already
decoded is reduced, and some such generations may become decodable
even when no new coded packets are received from the source. Again,
the newly decoded generations resolve some unknowns of the
generations they overlap with, which in turn may become decodable
and so on. This is the mechanism through which the generation
overlapping potentially improves the throughput. We declare
successful decoding when all $N$ information packets have been
decoded.

\paragraph{Computational Complexity} The computational complexity for
encoding is $\mathcal{O}((h+l)d)$ (recall that $d$ is the number of
information symbols in each packet as defined in Part {\em a)}) per
coded packet for linearly combining the $(h+l)$ information packets
in each generation. For decoding, the largest number of unknowns in
the systems of linear equations to be solved is $(h+l)$, and
therefore the computational complexity is upper bounded by
$\mathcal{O}((h+l)^2+(h+l)d)$ per information packet.

While some may argue that the assumption of random scheduling of
generations deviates from reality, we put forward here a few motives
behind its adoption: (1) Locality: Uniformly random scheduling
assumes knowledge of least information, which to some extent
approximates the case with limited-visioned peer nodes in
large-scale systems; (2) Ratelessness: in the case of single-hop
multicast over erasure channels, the code throughput automatically
adapts to all erasure rates. Some previous works on coding over
generations, such as \cite{petarchunked} and \cite{carletonoverlap},
also assumed random scheduling.

We measure code throughput by the number of coded packets necessary
for decoding all the information packets.


\subsection{Coupon Collector's Problem} As in \cite{nonOverlapping},
we model the collection of coded packets from $n$ generations as the
sampling of a set of $n$ coupons without replacement. In the next
section, we will look into the overlapping structure of the extended
generations, and use our extension of the \emph{coupon collector's
brotherhood} problem\cite{doubledixiecup,brotherhood} described in
\cite{nonOverlapping} to evaluate the throughput performance of the
random annex code. We will also compare the performance of random
annex code to the \emph{overlapped chunked code} proposed in
\cite{carletonoverlap}, which has the generations overlap in a
``head-to-toe'' fashion. Note that the random annex code in fact
defines a code ensemble that encompasses the overlapped chunked
code.

\section{Overlapping Generations-An Analysis of Expected Performance}
\label{sec:analysis}
\subsection{Overlapping Structure}\label{subsec:analysis_structure}
The decoding of different generations becomes intertwined with each
other as generations are no longer disjoint. Our goal here is to
unravel the structure of the overlapping generations, in order to
identify the condition for successful decoding of random annex codes
over a unicast link.

In Claims \ref{thm:pi} through \ref{thm:overlap_gennum}, we study
the overlapping structure of the random annex code. Compared with
the head-to-toe overlapping scheme, an extended generation in the
random annex code overlaps more evenly with other generations.
Intuitively, this can help with code throughput when random
scheduling of generations is used.
\begin{claim}\label{thm:pi}
For any packet in a base generation $B_{k}$, the probability that it
belongs to annex $R_r$ for some $r\in\{1,2,\dots,n\}\backslash\{k\}$
is
\[\pi={{N-h-1}\choose{l-1}}/{{N-h}\choose{l}}=\frac{l}{N-h}=\frac{l}{(n-1)h},\]
while the probability that it does not belong to $R_r$ is $\bar{\pi}
= 1-\pi.$
\end{claim}

\begin{claim}\label{thm:representation}
Let $X$ be the random variable representing the number of
generations an information packet participates in. Then, $X=1+Y$
where $Y$ is ${\rm Binom}(n-1,\pi).$
\[E[X]=1+(n-1)\pi=1+\frac{l}{h},\] and \[Var[X]=(n-1)\pi\bar{\pi}.\]
\end{claim}

\begin{claim} \label{thm:multiplerepresentation}
In any generation of size $g=h+l$, the expected number of
information packets not present in any other generation is
$h\bar{\pi}^{(n-1)}\approx he^{-l/h}$ for $n\gg1$. The expected
number of information packets present at least once in some other
generation is \[l+h[1-\bar{\pi}^{(n-1)}]\approx
l+h\left[1-e^{-l/h}\right]<\min\{g,2l\}\] for $n\gg1$ and $l>0$.
\end{claim}

\begin{claim}\label{thm:overlap_gennum}
The probability that two generations overlap is
$1-{{N-2h}\choose {l,l,N-2h-2l}}/{{N-h}\choose l}^2$. For
any given generation, the number of other generations it overlaps
with is then $${\rm Binom}\left(n-1,\left[1-{{N-2h}\choose
{l,l,N-2h-2l}}/{{N-h}\choose l}^2\right]\right).$$
\end{claim}

The following Theorem ~\ref{thm:union_overlap} gives the expected
overlap size $\Omega(s)$ between the union of $s$ generations and an
$(s+1)$th generation.
\begin{theorem}\label{thm:union_overlap}
For any $I\subset\{1,2,\dots,n\}$ with $|I|=s,$ and any $j\in
\{1,2,\dots,n\}\backslash I,$
\begin{align}
\Omega(s)=E[|\left(\cup_{i\in I} G_i\right)\cap
G_j|]&=
g\cdot \left[1-\bar{\pi}^s\right] +  s h \cdot \pi \bar{\pi}^s
\label{eq:union_overlap}\end{align}
 where $|B|$ denotes the cardinality of
set $B$, and $\pi$, $\bar{\pi}$ are as defined in Claim
\ref{thm:pi}.

When $n\rightarrow\infty$, if $\frac{l}{h}\rightarrow\alpha$ and
$\frac{s}{n}\rightarrow\beta$, and let $\omega(\beta)=\Omega(s)$,
then,
$\omega(\beta)\rightarrow h\left[(1+\alpha) \left(1 - e^{-\alpha\beta}\right)
+\alpha\beta e^{-\alpha\beta}\right].$
\end{theorem}
We provide a proof of Theorem~\ref{thm:union_overlap} in Appendix
\ref{app:union_overlap}.

\subsection{An Analysis of Overhead Based on Mean Values}
We next describe an analysis of the expected number of coded packets
a receiver needs to collect in order to decode all $N$ information
packets of $\mathcal{F}$ when they are encoded by the random annex
code. We base our analysis on Theorem \ref{thm:union_overlap} above
and Claim \ref{thm:wait} and Theorem \ref{thm:gen_wait} below, and
use the mean value for every quantity involved.

By the time when $s(s=0,1,\dots,n-1)$ generations have been decoded,
for any one of the remaining $(n-s)$ generations, on the average
$\Omega(s)$ of its participating information packets have been
decoded, or $(g-\Omega(s))$ of them are not yet resolved. If the
coded packets collected from some generation are enough for decoding
its unresolved packets, that generation becomes the $(s+1)$th
decoded one; otherwise, if no such generation exists, decoding
fails.


The following Claim \ref{thm:wait} estimates the number of coded
packets needed from a generation for its decoding when $(g-x)$ of
its information packets remain to be resolved.




\begin{claim}\label{thm:wait}
For any generation $G_i$, if $x$ of the $g=h+l$ information
packets of $G_i$ have been resolved by decoding other generations,
then, 
the expected value of the number of coded packets $N_i(g,x)$ from
$G_i$ needed to decode the remaining $(g-x)$ information packets
\begin{equation}\label{eq:wait}
E[N_{i}(g,x)]\lessapprox
g-x+\frac{q^{-1}}{1-q^{-1}}+\log_q\frac{1-q^{-(g-x)}}{1-q^{-1}}=\eta(x).
\end{equation}
\end{claim}
\noindent \begin{proof} In all the coding vectors, remove the
entries corresponding to the information packets already resolved.
Now we need to solve a system of linear equations of $(g-x)$
unknowns with all coefficients chosen uniformly at random from
$GF(q)$. Thus, $N_i(g,x)=N_i(g-x,0)$.
\begin{align}
E[N_{i}(g,x)]&=\sum_{j=0}^{g-x-1}\left(\frac{q^{g-x}-q^{j}}{q^{g-x}}\right)^{-1}\\
&=\int_{0}^{g-x-1} \frac{1}{1-q^{y-g+x}} dy + \frac{1}{1-q^{-1}} \\
&\lessapprox
g-x+\frac{q^{-1}}{1-q^{-1}}+\log_q\frac{1-q^{-(g-x)}}{1-q^{-1}}
\end{align}
\end{proof}


Extending the domain of $\eta(x)$ from integers to real numbers, we
can estimate that the number of coded packets needed for the
$(s+1)$th decoded generation should exceed $m'_s =\lceil
\eta(\Omega(s))\rceil$. Since in the random annex code, all
generations are randomly scheduled with equal probability, for
successful decoding, we would like to have at least $m'_0$ coded
packets belonging to one of the generations, at least $m'_{1}$
belonging to another, and so on. We wish to estimate the total
number of coded packets needed to achieve the above.

For any $m\in\mathbb{N}$, we define $S_m(x)$ as follows:
\begin{align}\label{eq:sm_m}
S_m(x)=&1+\frac{x}{1!}+\frac{x^2}{2!}+\dots+\frac{x^{m-1}}{(m-1)!}\quad(m\ge
1)\\
\label{eq:sm_0}S_{\infty}(x)=& \exp(x) ~\text{and} ~ S_{0}(x)=0.
\end{align}
The following Theorem~\ref{thm:gen_wait}, which is a restatement of
Theorem 2 from \cite{nonOverlapping} using the terminology of coding
over generations, provides a way to estimate the expected number of
coded packets necessary for decoding the whole file $\mathcal{F}$.

\begin{theorem}(Theorem 2, \cite{nonOverlapping})\label{thm:gen_wait} Suppose for some $A\in\mathbb{N}$, integers $k_1,\dots,k_A$ and
$m_1,\dots,m_A$ satisfy $1\le k_1<\dots<k_A\le n$ and
$\infty=m_0>m_1>\dots> m_A>m_{A+1}=0$. Let $\mu_{r}$ be the number
of generations for which at least $r$ coded packets have been
collected. Then, the expected number of coded packets necessary to
simultaneously have $\mu_{m_j}\ge k_j$ for all $j=1,2,\dots,A$ is
\begin{align}\label{eq:gen_wait}
&n\int_{0}^{\infty}\Bigl\{e^{nx}-\\
&\sum_{{{(i_1,i_2,\dots,i_A):\atop i_0=0,i_{A+1}=n}\atop k_j\le
i_j\le i_{j+1}}\atop j=1,2,\dots,A}\!\!
\prod_{j=0}^{A}{{i_{j+1}}\!\!\choose{i_j}}\Bigl[S_{m_{j}}(x)-S_{m_{j+1}}(x)\Bigr]^{i_{j+1}-i_{j}}\Bigr\}e^{-nx}dx\notag
\end{align}
\end{theorem}

A practical method to evaluate (\ref{eq:gen_wait}) is provided
in Appendix \ref{app:eval}. The computational complexity for one
evaluation of the integrand is $\mathcal{O}(An^2)$ given
$m_1=\mathcal{O}(An^2)$.
\\

The algorithm for our heuristic analysis is listed as follows:
\begin{enumerate}
\item Compute $\Omega(s-1)$ for $s=1,2,\dots,n$ using Theorem \ref{thm:union_overlap};
\item Compute $m_{s}^{\prime}=\lceil \eta(\Omega(s-1))\rceil$ for
$s=1,2,\dots,n$ using (\ref{eq:wait}) from Claim \ref{thm:wait};
\item \label{enum:cat}Map $m_s^{\prime}(s=1,2,\dots,n)$ into $A$ values $m_j(j=1,2,\dots,A)$ so that
$m_j=m_{k_{j-1}+1}^{\prime}=m_{k_{j-1}+2}^{\prime}=\dots=m_{k_{j}}^{\prime}
$, for $j=1,2,\dots,A$, $m_1>m_2>\dots>m_A>m_{A+1}=0$, $k_0=0$ and
$k_A=n$;
\item Evaluate (\ref{eq:gen_wait}) in Theorem \ref{thm:gen_wait} with
the $A$, $k_j$s, and $m_j$s obtained in Step \ref{enum:cat}), as an
estimate for the expected number of coded packets needed for
successful decoding.
\end{enumerate}
\begin{remark}
The above Step \ref{enum:cat}) is viable because $\Omega(s)$ is
nondecreasing in $s$, the righthand side of (\ref{eq:wait}) is
non-increasing in $x$ for fixed $g$, and thus $m_s^{\prime}$ is
non-increasing in $s$.
\end{remark}

Although our analysis is heuristic, we will see in the next section
that the estimate closely follows the simulated average performance
curve of the random annex coding scheme.

\section{Numerical Evaluation and Simulation Results}
\label{sec:results}
\subsection{Throughput vs.\ Complexity in Fixed Number of Generations Schemes}
Our goal here is to find out how the annex size $l$ affects the throughput performance
of the scheme with fixed base generation size $h$ and the total
number of information packets $N$ (and consequently,
the number of generations $n$). Note that we may be trading throughput for complexity
since the generation size $g=h+l$ affects the computational complexity of the scheme.

Figure \ref{fig:fixnh} shows the analytical and simulation results
when the total number $N$ of information packets is $1000$ and the
base generation size $h$ is $25$. Figure
\ref{fig:fixnh}\subref{subfig:fixhdist} shows $h+l-\Omega(s)$ for
$s=0,1,\dots,n$ with different annex sizes. Recall that $\Omega(s)$
is the expected size of the overlap of the union of $s$ generations
with any one of the rest $n-s$ generations. After the decoding of
$s$ generations, for any generation not yet decoded, the expected
number of information packets that still need to be resolved is then
$h+l-\Omega(s)$. We observe that the $h+l-\Omega(s)$ curves start
from $h+l$ for $s=0$ and gradually descends, ending somewhere above
$h-l$, for $s=n-1$.

Recall that we measure throughput by the number of coded packets
necessary for successful decoding. Figure
\ref{fig:fixnh}\subref{subfig:fixhexp} shows the expected
performance of the random annex code, along with the head-to-toe
overlapping code and the non-overlapping code($l=0$). Figure
\ref{fig:fixnh}\subref{subfig:fixhpe} shows the probability of
decoding failure of these codes versus the number of coded packets
collected.

\begin{figure}[h]
\begin{center}
\subfigure[]{\label{subfig:fixhdist}\includegraphics[scale=0.45]{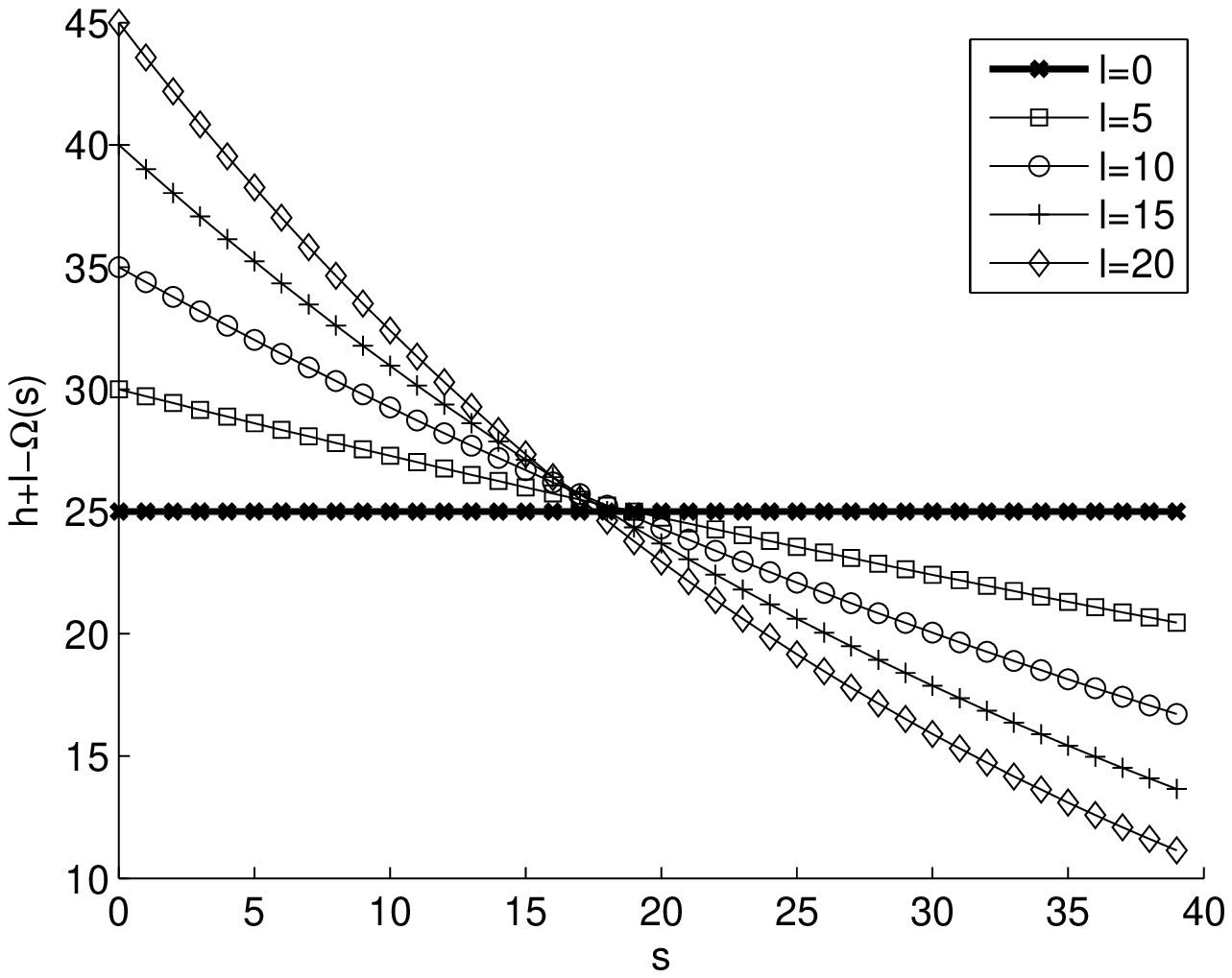}}\qquad
\subfigure[]{\label{subfig:fixhexp}\includegraphics[scale=0.5]{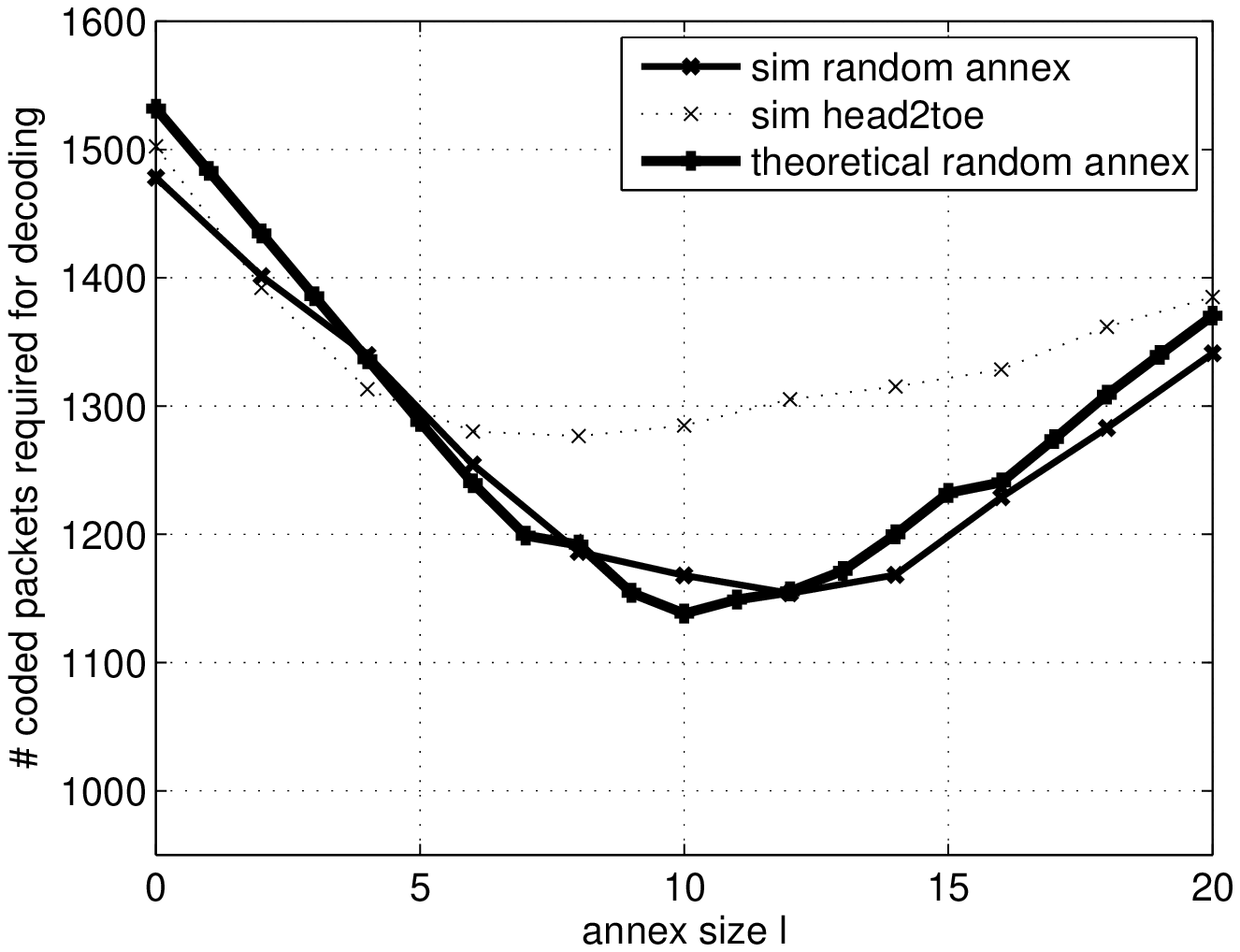}}\qquad
\subfigure[]{\label{subfig:fixhpe}\includegraphics[scale=0.5]{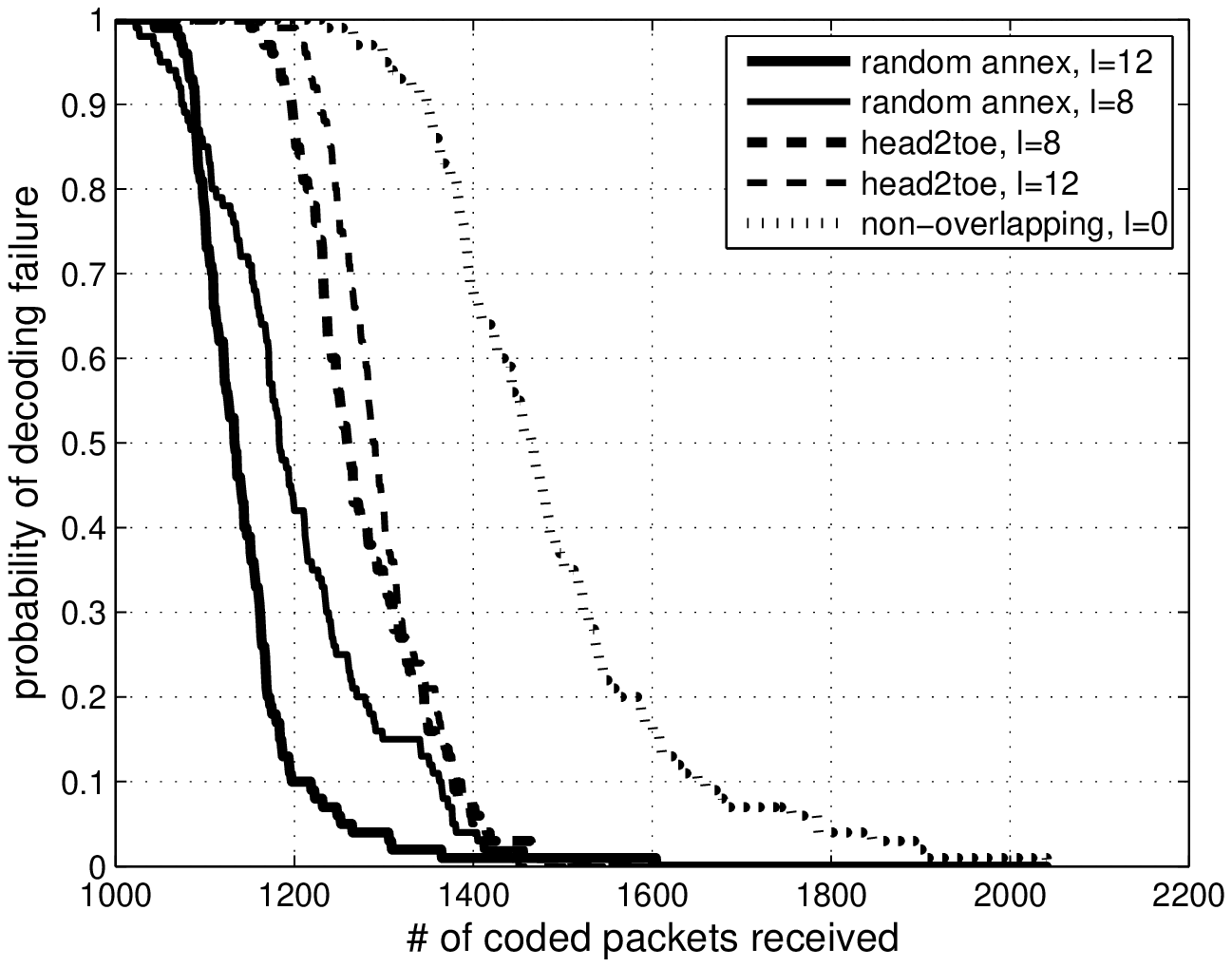}}
\caption{$N=1000$, $h=25$: \subref{subfig:fixhdist} $h+l-\Omega(s)$;
\subref{subfig:fixhexp} Expected Number of Coded Packets Needed for
Successful Decoding Versus Annex Size $l$; \subref{subfig:fixhpe}
Probability of Decoding Failure }\label{fig:fixnh}
\end{center}
\end{figure}
\begin{itemize}
\item Our analysis for the expected number of coded packets required
for successful decoding extremely closely matches the simulation results.
\item For both the random annex scheme and the head-to-toe
scheme, there is an optimal annex size, beyond or below which
throughput is lower than optimal. From the simulation results in
Figure \ref{fig:fixnh}\subref{subfig:fixhexp}, it is observed that
the optimal annex size is $12$ for the random annex scheme and $8$
for the head-to-toe scheme. Beyond the optimal annex size,
throughput cannot be increased by raising computational cost.
\item The random annex code outperforms head-to-toe overlapping at
their respective optimal points. Both codes outperforms the
non-overlapping scheme.
\item We also plotted the probability of decoding failure versus the
number of coded packets received. The probability of decoding
failure of the random annex code converges faster than those of the
head-to-toe and the non-overlapping scheme.
\end{itemize}

\subsection{Enhancing Throughput in Fixed Complexity Schemes}
Our goal here is to see how we can choose the annex size that
optimizes the throughput with negligible sacrifice in complexity.
To this end, we fix the extended generation size $g=h+l$ and vary only the annex size $l$.
Consequently, the computational complexity for coding remains roughly constant (actually
decreases with growing $l$).

Figure \ref{fig:fixgN} shows the analytical and simulation results
for the code performance when the total number $N$ of information
packets is fixed at $1000$ and size $g$ of extended generation fixed
at $25$.
\begin{figure}[h]
\begin{center}
\subfigure[]{\label{subfig:fixgexp}\includegraphics[scale=0.5]{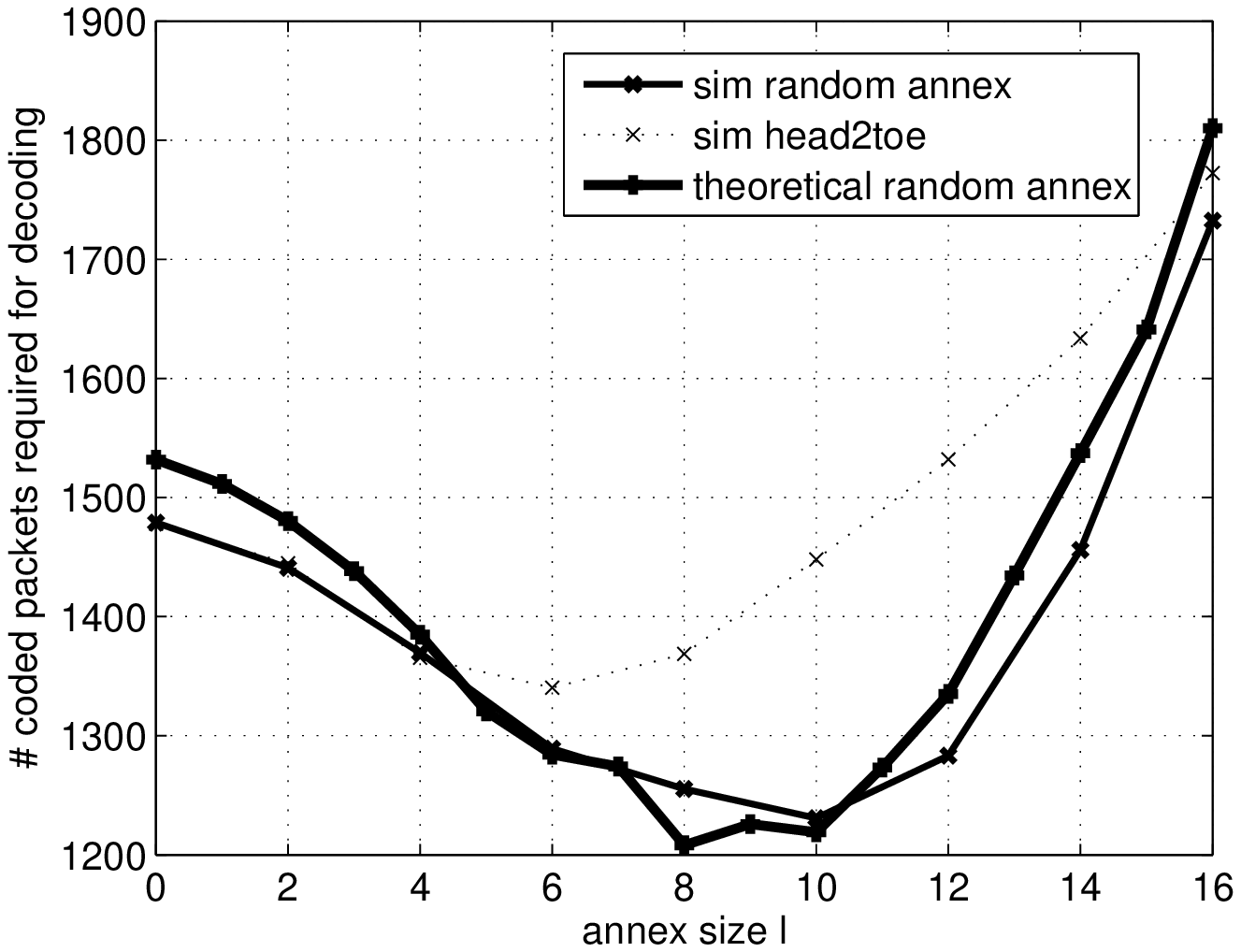}}\qquad
\subfigure[]{\label{subfig:fixgpe}\includegraphics[scale=0.5]{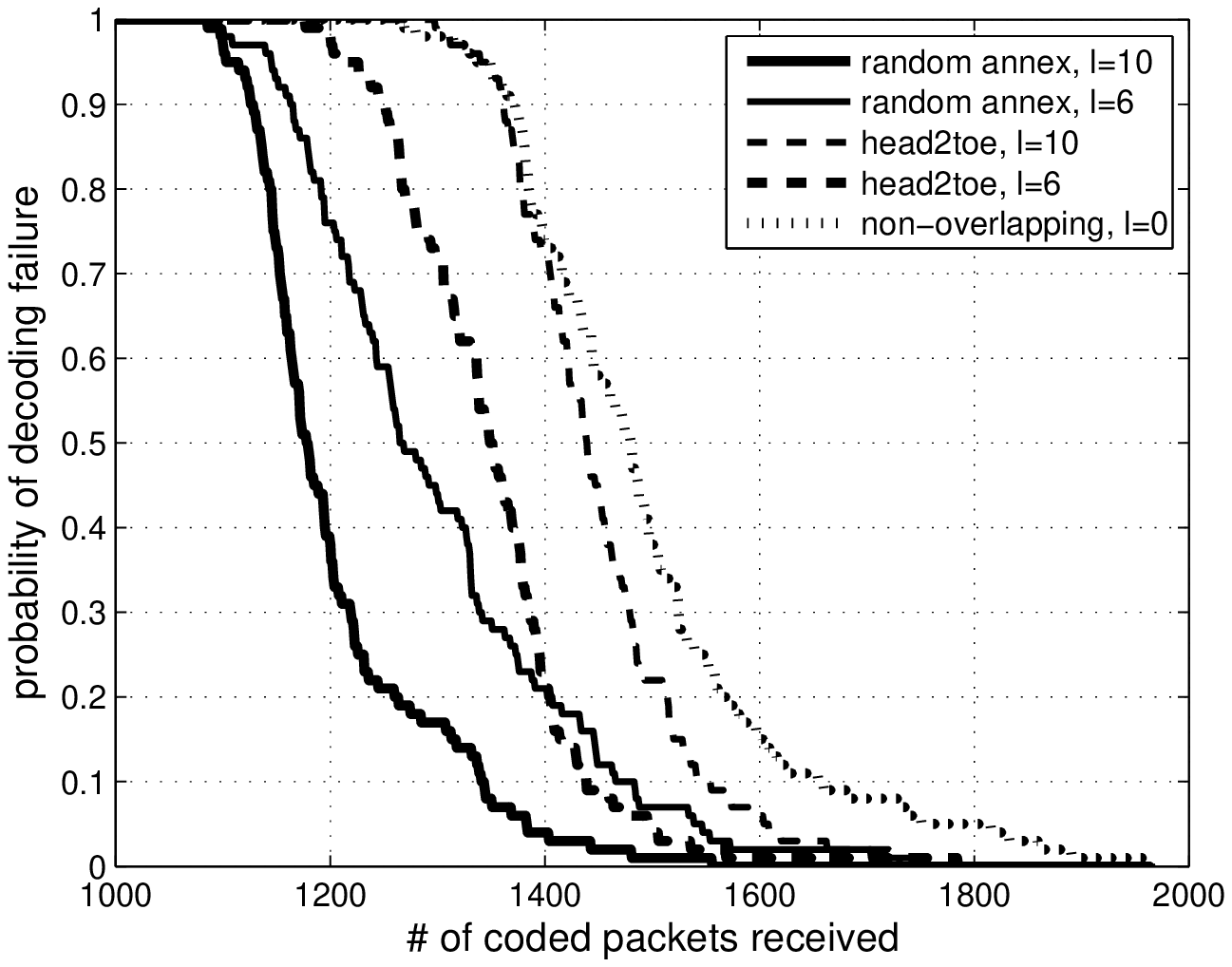}}
\caption{$N=1000$, $g=h+l=25$: \subref{subfig:fixgexp} Expected
Number of Coded Packets Needed for Successful Decoding Versus Annex
Size $l$; \subref{subfig:fixgpe} Probability of Decoding Failure
}\label{fig:fixgN}
\end{center}
\end{figure}

\begin{itemize}
\item Again our analytical results agree with simulation results very well;
\item It is interesting to observe that, without raising computational complexity, increasing annex size properly can still
give non-negligible improvement to throughput.  There is still an
optimal annex size that achieves highest throughput. From Figure
\ref{fig:fixgN}\subref{subfig:fixgexp}, we see that the optimal
annex size is $10$ for the random annex scheme and $6$ for the
head-to-toe scheme;
\item The random annex code again outperforms head-to-toe overlapping at
their optimal points. Both codes outperform the non-overlapping
scheme;
\item We again observe that the probability of decoding
failure of the random annex code converges faster than those of the
head-to-toe and the non-overlapping schemes.
\end{itemize}

\section{Conclusion And Future Work}\label{sec:conclusion}
We proposed a random annex scheme for coding with overlapping
generations. We obtained an accurate analytical evaluation of its
expected throughput performance for a unicast link using an
extension of the coupon collector's model derived in our recent work
\cite{nonOverlapping}. Both the expected throughput and the
probability of decoding failure of the random annex code are
generally better than those of the non-overlapping and head-to-toe
overlapping coding schemes. Under fixed information length and fixed
number of generations, there exists an optimal annex size that
minimizes the number of coded packets needed for successful
decoding. One of our most interesting findings is that when we fix
the information length and the generation size, increasing the annex
size may still improve throughput without raising computational
complexity.

We developed a practical algorithm to numerically evaluate some of
our complex analytical expressions. With slight modification of the
analytical method used in this work, we can also predict the
expected decoding progress, i.e., the number of
generations/information packets decodable with the accumulation of
coded packets. This will be useful to studies of content
distribution with tiered reconstruction at the user. It can also be
used to find the best rate of a ``precode'' \cite{petarchunked}
applied to coding over generations. It would be interesting to know
if the combination of overlapping generations and precode can
further improve code throughput.

It is also interesting to study the asymptotic performance of the
code, as the information length tends to infinity. We also hope to
characterize the optimal annex size in terms of generation size and
number of generations.

\appendices
\section{Proof of Theorem \ref{thm:union_overlap}}
\label{app:union_overlap}
 Without loss of
generality, let $I=\{1,2,\dots,s\}$ and $j=s+1$, and define
$\mathcal{R}_s = \cup_{i=1}^{s} R_i$, $\mathcal{B}_s=\cup_{i=1}^s
B_i$, and $\mathcal{G}_s = \cup_{i=1}^s G_i$ for $s=0,1,\dots,n-1$.
Then, $E\left[|\left(\cup_{i\in I} G_i\right)\cap
G_j|\right]=E\left[|\mathcal{G}_s \cap G_{s+1} |\right]$. For any
two sets $X$ and $Y$, we use $X+Y$ to denote $X\cup Y$ when $X\cap
Y=\emptyset$.
\begin{align*}
\mathcal{G}_s \cap G_{s+1} =& (\mathcal{B}_s + \mathcal{R}_s\backslash \mathcal{B}_s)\cap(B_{s+1}+R_{s+1})\\
=&\mathcal{B}_s\cap R_{s+1} + \mathcal{R}_s\cap B_{s+1} +
(\mathcal{R}_s\backslash \mathcal{B}_s)\cap R_{s+1}
\end{align*}
And therefore
\begin{align}\label{eq:part}
E[|\mathcal{G}_s \cap  G_{s+1}|]= & E[|\mathcal{B}_s\cap R_{s+1}|]+\\ &
E[|\mathcal{R}_s\cap B_{s+1}|] + E[|(\mathcal{R}_s\backslash
\mathcal{B}_s)\cap R_{s+1}|]\notag
\end{align}
Using Claim \ref{thm:pi}, we have
\begin{align}
&E[|\mathcal{B}_s\cap R_{s+1}|]=sh\pi,\label{eq:exp1}\\
&E[|\mathcal{R}_s\cap B_{s+1}|]=h[1-(1-\pi)^s],\label{eq:exp2}\\
&E[|(\mathcal{R}_s\backslash \mathcal{B}_s)\cap
R_{s+1}|]=(n-s-1)h\pi[1-(1-\pi)^s],\label{eq:exp3}
\end{align}
where $\pi$ is as defined in Claim \ref{thm:pi}.
 Bringing (\ref{eq:exp1})-(\ref{eq:exp3}) into (\ref{eq:part})
we obtain (\ref{eq:union_overlap}).




Furthermore, when $n\rightarrow\infty$, if $l/h\rightarrow\alpha$
and $s/n\rightarrow\beta$, then
\begin{align*}
E[|\mathcal{G}_s \cap G_{s+1}|] =& g\cdot \left[1-\bar{\pi}^s\right]
+  s h \cdot \pi \bar{\pi}^s\\
\rightarrow&h(1+\alpha)
\Bigl[1-\Bigl(1-\frac{\alpha}{n-1}\Bigr)^{n\beta}\Bigr]+\\&h\alpha\beta
\Bigl(1-\frac{\alpha}{n-1}\Bigr)^{n\beta}\notag\\
\rightarrow& h\Bigl[(1+\alpha)(1-e^{-\alpha\beta})+\alpha
\beta e^{-\alpha\beta}\Bigr]\notag\\
=& h\Bigl[1+\alpha -
(1+\alpha-\alpha\beta)e^{-\alpha\beta}\Bigr]\notag
\end{align*}

\section{Evaluation of Expression (\ref{eq:gen_wait})}
\label{app:eval} We give here a method to calculate the integrand in
(\ref{eq:gen_wait}) of Theorem \ref{thm:gen_wait}. The integrand of
(\ref{eq:gen_wait}) can be rewritten as
\begin{align}\label{eq:sums}
&1-\sum_{i_A=k_A}^{n}{n\choose
i_A}[(S_{m_A}(x)-S_{m_{A+1}}(x))e^{-x}]^{n-i_A}\\ \notag
&\cdot\sum_{i_{A-1}=k_{A-1}}^{i_A}{i_A\choose
i_{A-1}}[(S_{m_{A-1}}(x)-S_{m_{A}}(x))e^{-x}]^{i_A-i_{A-1}}\\
\notag &\dots\sum_{i_1=k_1}^{i_2}{i_2\choose
i_1}[(S_{m_{1}}(x)-S_{m_{2}}(x))e^{-x}]^{i_2-i_1}\\
\notag &\cdot{i_1\choose
i_0}[(S_{m_0}(x)-S_{m_1}(x))e^{-x}]^{i_1-i_0}.
\end{align}
For $k=k_{1}, k_{1}+1, \dots, n$, let
\begin{align*}
&\phi_{0,k}(x)=[(S_{m_0}(x)-S_{m_1}(x))e^{-x}]^{k};
\end{align*}
For each $j=1,2,\dots,A$, let
\begin{align*}
&\phi_{j,k}(x)=\sum_{w=k_j}^{k}{k\choose
w}[(S_{m_{j}}(x)-S_{m_{j+1}}(x))e^{-x}]^{k-w}\phi_{j-1,w}(x),\\
&\textnormal{for }  k=k_{j+1},k_{j+1}+1,\dots,n.
\end{align*}
Then, one can verify that (\ref{eq:sums}) is exactly
$1-\phi_{A,n}(x)$.

It is not hard to find an algorithm that calculates
$1-\phi_{A,n}(x)$ in
$(c_1m_1+c_2(n-1)+c_3\sum_{j=1}^A\sum_{k=k_{j+1}}^n(k-k_j))$ time,
where $c_1$, $c_2$ and $c_3$ are positive constants. As long as
$m_1=\mathcal{O}(An^2)$, we can estimate the amount of work for a
single evaluation of the integrand of (\ref{eq:sm_m}) in Theorem
\ref{thm:gen_wait} to be $\mathcal{O}(An^2)$.


The integral can be computed through the use of some efficient
quadrature method, for example, Gauss-Laguerre quadrature. For
reference, some numerical integration issues of the special case
where $A=1$ have been addressed in Part 7 of \cite{Flajolet1992207}
and in \cite{couponrevisited}.




\section*{Acknowledgements}
Special thanks to Meiyue Shao and Professor Richard S. Falk for
useful discussions concerning numerical analysis issues.


\bibliographystyle{IEEEtran}
\bibliography{netcodrefs}

\end{document}